\documentclass{extarticle}

\AtBeginDocument{%
  \providecommand\BibTeX{{%
    \normalfont B\kern-0.5em{\scshape i\kern-0.25em b}\kern-0.8em\TeX}}}

\usepackage[utf8]{inputenc}
\usepackage[english]{babel}
\usepackage{amsthm}
\usepackage{amsmath}
\usepackage{amssymb}
\usepackage{color}
\usepackage{tikz}
\usepackage{enumitem}
\usepackage{hyperref}
\usepackage[a4paper, margin=2.5cm]{geometry}
\usepackage[numbers,sort]{natbib}

\newtheorem{Problem}{Problem}[section]

\newtheorem{Theorem}[Problem]{Theorem}
\newtheorem{Lemma}[Problem]{Lemma}

\newtheorem{Algorithm}[Problem]{Algorithm}
\newtheorem{Proposition}[Problem]{Proposition}
\newtheorem{Definition}[Problem]{Definition}
\newtheorem{Example}[Problem]{Example}
\newtheorem{Notation}[Problem]{Notation}

\def\clap#1{\hbox to0pt{\hss#1\hss}}
\def\eatspace#1{#1}
\def\step#1#2{\par\kern1pt\dimen44=#2em\advance\dimen44 1.67em\hangindent\dimen44\hangafter=1\noindent\rlap{\small#1}\kern\dimen44\relax\eatspace}
\let\set\mathbb
\def\<#1>{\langle#1\rangle}
\def\K{\set K}

\DeclareMathOperator{\sep}{sep}

\DeclareMathOperator{\Gal}{Gal}
\DeclareMathOperator{\lp}{lp}

\DeclareMathOperator{\rem}{rem}

\begin{document}

\title{Separating Variables in Bivariate Polynomial Ideals}

\author{Manfred Buchacher\footnote{\href{mailto:manfred.buchacher@jku.at}{manfred.buchacher@jku.at}, Institute for Algebra, Johannes Kepler University, Linz, Austria},
Manuel Kauers\footnote{\href{mailto:manuel.kauers@jku.at}{manuel.kauers@jku.at}, Institute for Algebra, Johannes Kepler University, Linz, Austria},
Gleb Pogudin\footnote{
\href{mailto:gleb.pogudin@polytechnique.edu}{gleb.pogudin@polytechnique.edu}, Department of Computer Science, Higher School of Economics, Moscow, Russia and LIX, CNRS, \'Ecole Polytechnique, Institut Polytechnique de Paris, France
}}

\date{}

\maketitle

\begin{abstract}
  We present an algorithm which for any given
  ideal $I\subseteq\K [x,y]$ finds all elements of $I$
  that have the form $f(x)-g(y)$, i.e., all elements
  in which no monomial is a multiple of~$xy$.
\end{abstract}

\section{Introduction}

One of the fundamental problems in computer algebra and applied
algebraic geometry is the problem of elimination. Here, we are given a
polynomial ideal $I\subseteq\K[x_1,\dots,x_n,y_1,\dots,y_m]$ and the
task is to compute generators of the ideal $I\cap\K[x_1,\dots,x_n]$.
The resulting ideal of $\K[x_1,\dots,x_n]$ consists of
all elements of $I$ that do not contain any terms that are a multiple of
any of the variables~$y_i$. It is well-known that this problem can be
solved by computing a Gr\"obner basis with respect to an elimination
order that assigns higher weight to terms involving $y_1,\dots,y_m$ than
to terms not involving these variables.

It is less clear how to use Gr\"obner bases (or any other standard elimination techniques) for
finding ideal elements that do not contain any terms which are a
multiple of certain prescribed terms rather than certain prescribed
variables. The problem considered in this paper is an elimination
problem of this kind. Here, given an ideal
$I\subseteq\K[x_1,\dots,x_n,y_1,\dots,y_m]$, we are interested in all
elements of $I$ that do not involve any terms which are multiples of any
of the terms $x_iy_j$ ($i=1,\dots,n$, $j=1,\dots,m$). Note that, these
are precisely the elements of $I$ which can be written as the sum of a
polynomial in $x_1,\dots,x_n$ only and a polynomial in $y_1,\dots,y_m$
only, so the problem under consideration is as follows.

\begin{Problem}[Separation]
\begin{description}
  \item[]
  \item[Input] An ideal $I \subseteq \K[x_1, \ldots, x_n, y_1, \ldots, y_m]$;
  \item[Output] Description of all $f - g \in I$ such that 
  \[
  f \in \K[x_1, \ldots, x_n] \text{ and } g \in \K[y_1, \ldots, y_m].
  \]
\end{description}
\end{Problem}

At first glance, it may seem that there should be a simple way to solve
this problem with Gr\"obner bases, similarly as for the classical
elimination problem. However, we were not able to come up with such an
algorithm. The obstruction seems to be that there is no term order that
ranks the term $xy$ higher than both $x^2$ and~$y^2$.

We ran into the need for such an algorithm when we tried to automatize an interesting
non-standard elimination step which appears in Bousquet-M\'elou's ``elementary''
solution of Gessel's walks~\cite{bousquet16}. Dealing with certain power series, 
say $u\in\K [x][[t]]$ and $v\in\K [x^{-1}][[t]]$, she finds polynomials $f,g$ such 
that $f(u)-g(v)=0$, and then concludes that $f(u)$ and $g(v)$ must
in fact belong to $\K [[t]]$. Deriving a pair $(f,g)$ automatically from known relations
among $u,v$ amounts to the problem under consideration.

The problem also arises when one wants to compute the intersection of two $\K $-algebras.
For example, suppose that for given $u,v\in\K [t_1,\dots,t_n]$ one wants to compute
$\K [u]\cap\K [v]$. This can be done by finding all pairs $(f,g)$ such that $f(u)=g(v)$, i.e., all pairs $(f,g)$ with $f(x)-g(y)\in\<x-u,y-v>\cap\K [x,y]$. See~\cite{beals2009,engstrom1941} for a discussion of this and similar problems.

\begin{Definition}
 Let $p\in\K [x_1,\dots,x_n,y_1,\dots,y_m]$.
 \begin{enumerate}
 \item $p$ is called \emph{separated} if there exist
   $f\in\K [x_1,\dots,x_n]$ and $g\in\K [y_1,\dots,y_m]$ such that 
   $p=f-g$.
 \item $p$ is called \emph{separable} if there is a
   $q\in\K [x_1,\dots,x_n,y_1,\dots,y_m]$ such that $qp$ is separated.
 \end{enumerate}
\end{Definition}

\begin{Proposition}
  Let $I$ be an ideal in $\K [x_1,\dots,x_n,y_1,\dots,y_m]$.
  Then 
    \[
      A(I):=\{\,(f,g)\in\K [x_1,\dots,x_n]\times\K [y_1,\dots,y_m] : f - g\in I\,\}
    \]
   is a unital $\K $-algebra with respect to component-wise addition and multiplication and component-wise multiplication by elements of $\K$. We refer to $A(I)$ as the \emph{algebra of separated polynomials} of~$I$.
\end{Proposition}

\begin{proof}
We just note that $A(I)$ is clearly a $\K $-vector space, and that 
it is closed under component-wise multiplication, as for any $(f,g),(f',g')\in A(I)$ we have
$f-g\in I$ and $f'-g'\in I$, so $(f-g)f'+g(f'-g')=ff'-gg'\in I$. It is unital, because we always have $(1,1)\in A(I)$. 
\end{proof}

Given ideal generators of~$I$, we want to determine $\K $-algebra generators of~$A(I)$.
This is in general too much to be asked for, because, as shown in Example~\ref{ex:not_finitely_generated},
$A(I)$ may not be finitely generated.
On the positive side, it is known that $A(I)$ is finitely generated if $I$ is a principal ideal in the ring of bivariate polynomials (see \cite{fried69}). 

The main result of the paper is Algorithm~\ref{alg:general} for computing
generators of the algebra $A(I)$ for a given bivariate ideal $I \subseteq \K[x, y]$.
In particular, it implies that such an algebra is always finitely generated and yields an algorithm to compute a minimal separated multiple of a bivariate polynomial~\cite[Definition~4.1]{fried69}.
An implementation of the algorithm in Mathematica can be found on the website of the second author.

The general structure of the algorithm is the following. 
Every bivariate ideal is the intersection of a zero-dimensional ideal
and a principal ideal.
We solve the separation problem for the zero-dimensional case (Section~\ref{sec:0}) and for the principal case (Section~\ref{sec:principal}) separately.
Then we show how to compute the intersection of the resulting algebras in Section~\ref{sec:arbitrary}.
We conclude with discussing the case of more than two variables in Section~\ref{sec:more_than_two}.

In the context of separated polynomials, many deep results have been obtained
for some kind of ``inverse problem'' to the problem considered here, i.e., the study
of the shape of factors of polynomials of the form $f(x) - g(y)$, 
see~\cite{BT00, B99, CNC99, DLS61, F73, cassels, fried69} and references therein.
We use techniques developed in~\cite{cassels} in our proofs (see Section~\ref{sec:principal}).

We assume throughout that the ground field $\K$ has characteristic zero and that
for a given element of an algebraic extension of $\K$ we can decide whether it is
a root of unity. 
This is true, for example, for every number field (see Section~\ref{subsec:algorithm}).

It is an open question whether the assumption on the characteristic of~$\K$ can be eliminated.
In positive characteristic, additional phenomena have to be taken into account.
For example, separable polynomials need not be squarefree, as the example $(x+y)^2\in\set Z_3[x,y]$
shows, which is separable because $(x+y)(x+y)^2=(x+y)^3=x^3+y^3$.


\section{Zero-Dimensional Ideals}\label{sec:0}

When $I\subseteq\K [x,y]$ has dimension zero, it is easy to separate variables.
In this case, there are nonzero polynomials $p,q$ with $I\cap\K [x]=\<p>$ and $I\cap\K [y]=\<q>$.
Clearly, these univariate polynomials $p$ and $q$ are separated.
Also all $\K [x]$-multiples of $p$ and all $\K [y]$-multiples of~$q$ are
separated elements of~$I$.

An arbitrary pair $(f,g)\in\K [x]\times\K [y]$ belongs to $A(I)$ if and only if
$(f+up,g+vq)$ belongs to $A(I)$ for all $u\in\K [x]$ and $v\in\K [y]$. In particular,
we have $(f,g)\in A(I)\iff (\rem_x(f, p),\rem_y(g, q))\in A(I)$. It is therefore sufficient to find
all pairs $(f,g)\in A(I)$ with $\deg_x f<\deg_x p$ and $\deg_y g<\deg_y q$. These pairs can be found
with linear algebra.

\begin{Algorithm}\label{alg:1}
  Input: $I\subseteq\K [x,y]$ of dimension zero.\\
  Output: generators of the $\K $-algebra $A(I)\subseteq\K [x]\times\K [y]$

  \step 10 if $I=\<1>$, return $\{(1,0),(x,0),(0,1),(0,y)\}$.
  \step 20 compute $p\in\K [x]$ and $q\in\K [y]$ such that
  \[
    I\cap\K [x]=\<p>\;\text{ and }\;I\cap\K [y]=\<q>.
  \]
  \step 30 make an ansatz $h=\sum_{i=0}^{\deg_x p-1} a_ix^i - \sum_{j=0}^{\deg_y q-1}b_jy^j$
  with undetermined coefficients $a_i,b_j$.
  \step 40 compute the normal form of $h$ with respect to a Gr\"obner basis of~$I$ and
  equate its coefficients to zero.
  \step 50 solve the resulting linear system over $\K $ for the unknowns $a_i,b_j$ and
  let $(f_1,g_1),\dots,(f_d,g_d)$ be the pairs of polynomials corresponding to a basis of
  the solution space.
  \step 60 return $(f_1,g_1),\dots,(f_d,g_d),(p,0),\dots,(x^{\deg_x p-1}p,0)$,
  $(0,q),\dots,(0,y^{\deg_y q-1}q)$.
\end{Algorithm}

\begin{Proposition}\label{prop:alg_zerodim_correct}
  Algorithm~\ref{alg:1} is correct.
\end{Proposition}
\begin{proof}
  It is clear by construction that all returned elements belong to~$A(I)$. It remains to show
  that they generate $A(I)$ as $\K $-algebra.
  This is clear if $I=\<1>$, because then $A(I)=\K[x]\times\K[y]$.
  Now suppose that $I\neq\<1>$ and let $(f,g)\in A(I)$. 
  Because of $I\neq\<1>$, we have $\deg_x p,\deg_y q>0$.
  Then $\<p>\subseteq\K [x]$ is generated as a $\K $-algebra by 
  $p,xp,\dots,x^{\deg_x p-1}p$.
  To see this, we just note that, by performing repeatedly division by $p$ on a polynomial and the resulting quotients, 
  any $u\in\<p>$ can be written 
  $$
  u = \sum_{i=1}^k r_i p^i
  $$
  where $r_i$ are polynomials with $\deg r_i < \deg p$. 
  Hence, $\<p>$ is a subset of the algebra generated by $p,xp,\dots,x^{\deg_x p-1}p$, and clearly, the reverse inclusion holds as well.
  For the same reason, $\<q>$ is generated as $\K$-algebra by $q,xq,\dots,x^{\deg_x q-1}q$.
  
  Hence $(f,g)$ can be expressed in terms of the given generators if and only if 
  $(\rem_x(f,p),\rem_y(g,q))$ can be expressed in terms of the given generators.
  Because of $\deg_x(\rem_x(f,p))<\deg_x(p)$ and $\deg_y(\rem_y(g,q))<\deg_y(q)$, the pair
  $(\rem_x(f,p),\rem_y(g,q))$ is a $\K $-linear combination of $(f_1,g_1),\dots,(f_d,g_d)$, as required. 
\end{proof}

\begin{Example}
 Consider the $0$-dimensional ideal $I=\<x^2y^2-1,y^5+y^3+xy^2+x>$.
 We have
 \[
   I\cap\K[x]=\<x^{10}+x^8-x^2-1>
   \kern.5em\text{and}\kern.5em
   I\cap\K[y]=\<y^{10}+y^8-y^2-1>.
 \]
 Every separated polynomial of $I$ therefore has the form
 \[
   f(x) + u(x)(x^{10}+x^8-x^2-1)
   -g(y) - v(y)(y^{10}+y^8-y^2-1)
 \]
 for certain $f(x),g(y)$ of degree less than $10$ and some $u(x),v(y)$.
 To find the pairs $(f,g)$, compute the normal form of 
 $h = \sum_{i=0}^9 a_ix^i - \sum_{i=0}^9b_jy^j$ with respect to a Gr\"obner
 basis of~$I$. Taking a degrevlex Gr\"obner basis, this gives
 \begin{alignat*}1
  (a_0+a_8-b_0) + (a_6-b_2)y^2 + (a_7+b_5)xy^2+\cdots.
 \end{alignat*}
 Equate the coefficients with respect to $x,y$ to zero and solve the 
 resulting linear system for the unknowns $a_0,\dots,a_9,b_0,\dots,b_9$.
 The following pairs of polynomials $(f,g)$ correspond to a basis of the
 solution space:
 \begin{alignat*}1
  &(1,1),\ (x-x^9,y^9-y),\ (x^2,y^8+y^6-1), \ (x^9+x^3,-y^9-y^3)\\
  &(x^4,-y^8+y^4+1),\ (x^5-x^9,y^3-y^7),\ (x^6,y^8+y^2-1)\\
  &(x^9+x^7,-y^5-y^3), \ (x^8,2-y^8).
 \end{alignat*}
 These pairs together with the pairs $(x^i(x^{10}+x^8-x^2-1),0)$
 and $(0,y^i(y^{10}+y^8-y^2-1))$ for $i=0,\dots,9$ form a set of
 generators of~$A(I)$. 
\end{Example}

For an ideal $I\subseteq\K [x,y]$ to be zero-dimensional means that its
codimension as $\K $-subspace of~$\K [x,y]$ is finite. Note that, in this
case, also $A(I)$ has finite codimension as $\K $-subspace of $\set
K[x]\times\K [y]$. Since we will need this feature later, let us record it as a lemma.

\begin{Lemma}\label{lemma:1}
  If $I\subseteq\K [x,y]$ has dimension zero, then there is a finite-dimensional
  $\K $-subspace $V$ of $\K [x]\times\K [y]$ such that the direct sum $V\oplus A(I)$ is equal to $\K [x]\times\K [y]$. Moreover, we can compute a basis of such a~$V$, and for
  every $(f,g)\in\K [x]\times\K [y]$ we can compute a $(\tilde f,\tilde g)\in V$
  such that $(f,g)-(\tilde f,\tilde g)\in A(I)$.
\end{Lemma}
\begin{proof}
  Let $p,q,(f_1,g_1),\dots,(f_d,g_d)$ be as in Algorithm~\ref{alg:1}.
  Note that, as a $\K $-vector space, $A(I)$ has the basis
  \[
    \{(f_1,g_1),\dots,(f_d,g_d)\}\cup\{(x^k p,0):k\in\set N\}\cup\{(0,y^k q):k\in\set N\}.
  \]
  Using row-reduction, it can be arranged that the $f_i$ have pairwise
  distinct degrees. Note that, all $f_i$ are nonzero by the choice of~$q$.
  Let $V$ be the $\K $-subspace of $\K [x]\times\K [y]$ generated by the
  pairs $(x^k,0)$ for all $k<\deg_x(p)$ which are not the degree of some~$f_i$
  and the pairs $(0,y^k)$ for all $k<\deg_y(q)$. We have $V\oplus A(I)=\K [x]\times\K [y]$.

  Given $(f,g)\in\K [x]\times\K [y]$, we compute $(\rem_x(f,p),\rem_y(g,q))$,
  and then eliminate all terms from the first component whose exponent is the degree of an~$f_i$. 
  The resulting pair $(\tilde f,\tilde g)$ is an element of $V$ with $(f,g)-(\tilde f,\tilde g)\in A(I)$.
\end{proof}


\section{Principal Ideals}\label{sec:principal}

We now consider the case where $I=\<p>$ is a principal ideal of $\K [x,y]$.
If $p\in\K[x]\cup\K[y]$, the algebra $A(I)$ of separated polynomials is finitely generated, 
as we have seen in the proof of Proposition~\ref{prop:alg_zerodim_correct}. 
It was shown in~\cite[Theorem 4.2]{fried69} 
that, if $p$ is separable, there is a separated multiple $f(x) - g(y)$ of $p$ that divides any other separated 
multiple of it. 
We refer to $f(x) - g(y)$ as \emph{the minimal separated multiple} of~$p$.
Moreover, \cite[Theorem~2.3]{fried69}~implies that if $p \not\in \K[x] \cup \K[y]$, then $(f,g)$ is an algebra generator for~$A(I)$. 
We note that,~\cite[Theorem~2.3]{fried69} was reproven in~\cite{binder96}, and generalized further in~\cite{aichinger11,schicho95}.
The proof of~\cite[Theorem 4.2]{fried69} was not constructive.
In the following we provide a criterion that allows to decide if $p$ is separable, and if it is, to compute 
its minimal separated multiple.

Our criterion is based on considering the highest graded component of the polynomial
with respect to a certain grading.
The separability of the highest component is a necessary but not a sufficient condition 
for the separability of a polynomial itself.
Surprisingly, there is a weaker converse, that is, the minimal separated multiple 
of the highest component is equal to the highest component of the minimal separated 
multiple of $p$ \emph{if the latter exists} (see Theorem~\ref{thm:reduction_to_lp}).
This allows us to reduce the problem for a general not necessarily homogeneous polynomial to the same problem for a homogeneous polynomial (which is solved in Section~\ref{subsec:homogeneous}) 
and solving a linear system.
The resulting algorithm is presented in Section~\ref{subsec:algorithm}.

Since the case $p \in \mathbb{K}[x] \cup \mathbb{K}[y]$ is trivial, for the rest of the section, we assume that $p \in \mathbb{K}[x, y] \setminus (\mathbb{K}[x] \cup \mathbb{K}[y])$.


\subsection{Homogeneous case}\label{subsec:homogeneous}

\begin{Definition}\label{def:weight}
\begin{enumerate}
    \item[]
    \item A function $\omega$ from the set of monomials in $x$ and $y$ to $\mathbb{R}$
    is called a \emph{weight function} if there exist $\omega_x, \omega_y \in \mathbb{Z}_{>0}$
    such that $\omega(x^i y^j) = \omega_x i + \omega_y j$ for every $i, j \in \mathbb{Z}_{\geq 0}$.
    
    \item Two weight functions are considered to be \emph{equivalent} if they differ by a constant non-zero factor.
    
    \item For a weight function $\omega$ and a nonzero polynomial $p \in\K [x, y]$, $\omega(p)$ is defined to be the maximum of the weights of the monomials of $p$.
    
    \item For a weight function $\omega$ and a polynomial $p \in\K [x, y]$, we define 
    \emph{the $\omega$-leading part of $p$} (denoted by $\lp_{\omega}(p)$) as the sum 
    of the terms of $p$ of weight $\omega(p)$.
\end{enumerate}
\end{Definition}

In this subsection, we consider the case of $p$ being homogeneous with respect to some weight function $\omega$, that is, $\lp_{\omega}(p) = p$.

\begin{Proposition}\label{prop:homog}
   Let $\omega$ be a weight function, and let $p \in \K[x, y] \setminus (\K[x] \cup \K[y])$ satisfy $\lp_{\omega}(p) = p$.
   Then $p$ is separable if and only if
   \begin{enumerate}
       \item $p$ involves a monomial only in $x$, and
       \item all the roots of $p(x, 1)$ in 
       the algebraic closure $\overline{\K}$ of $\K$ are distinct 
       and the ratio of every two of them is a root of unity.
   \end{enumerate}
   Moreover, if $p$ is separable and $N$ is the minimal number such that the ratio of every pair of roots of $p(x, 1)$ is an $N$-th root of unity, then the
   weight of the minimal separated multiple of $p$ is $N\omega_x$.
\end{Proposition}

\begin{proof}
   Assume that $p$ is separable, and let $P$ be a separated multiple.
   Replacing $P$ with $\lp_{\omega}(P)$ if necessary, we will further assume
   that $P = \lp_{\omega}(P)$.
   Since $P \notin \K[x] \cup \K[y]$ and is separated,
   $P$ involves a monomial in $x$ only, and hence, so does $p$.
   
   Since $P$ is $\omega$-homogeneous and separated, it is of the form $ax^m - by^n$ for some $a, b \in \K\setminus\{0\}$, 
   so $p(x, 1) \mid ax^m - b$.
   All roots of the latter are distinct and the ratio of each of them is an $m$-th root of unity.
   Hence, the same is true for $p(x, 1)$.
   This proves the only-if part of the proposition.
   
   To prove the remaining part of the proposition, 
   let $N$ be as in the statement of the proposition, 
   and $\gamma \in \overline{\K}$ be a root of $p(x, 1)$.
   Consider the $\omega$-homogeneous Puiseux polynomial
   \[
     P := x^N - \gamma^N y^{N\omega_x / \omega_y}.
   \]
   We perform Euclidean division of $P$ by $p$ over the field $F$ of Puiseux series in $y$ over $\overline{\K}$.
   This will yield a representation $P = qp + r$, where $q$ and $r$ are also $\omega$-homogeneous.
   Since $P(x, 1)$ is divisible by $p(x, 1)$, we see that $r(x, 1) = 0$.
   However, the $\omega$-homogeneity of $r$ implies that each of its coefficients with respect to $x$
   is a Puiseux monomial in~$y$.
   Thus, $r = 0$.
   Next, assume that $N\omega_x / \omega_y$ is not an integer. 
   Then there is an automorphism $\sigma$ of the Galois group of $F$ over $\overline{\K}(y)$ that moves $y^{N\omega_x / \omega_y}$.
   Then
   \[
   p \mid P - \sigma(P) \in F,
   \]
   which is impossible.
   Therefore, $P$ is a separated polynomial divisible by $p$ of weight $N\omega_x$.
\end{proof}

Of course, because of symmetry, the statements of Proposition~\ref{prop:homog} also hold for $y$ instead of $x$.


\subsection{Reduction to the homogeneous case}

We will start with a necessary condition for $p$ being separable.

\begin{Lemma}\label{lem:necessary_condition}
  Let $p \in \K[x, y] \setminus (\K[x] \cup \K[y])$ be separable.
  \begin{enumerate}
      \item There exists a unique (up to a constant factor) weight function $\omega$ such that $\lp_\omega(p)$ involves at least two monomials.
      
      \item The polynomial $\lp_\omega(p)$ is separable.
  \end{enumerate}
\end{Lemma}

\begin{proof}
  Let $q \in \K [x, y] \setminus \{0\}$ be such that $qp$ is separated.
  Let $\deg_x qp = m$ and $\deg_y qp = n$.
  Define $\omega(x^i y^j) = ni + mj$.
  If $\lp_{\omega}(p)$ contains only one monomial, then
  every monomial in $\lp_{\omega}(qp)$ is divisible by it.
  This is impossible since $\lp_{\omega}(qp)$ involves both $x^m$ and $y^n$.
  
  To prove the uniqueness, assume that there are two nonequivalent weight functions $\omega_1$ and $\omega_2$ with
  this property. 
  Since $\lp_{\omega_i}(qp) = \lp_{\omega_i}(q) \lp_{\omega_i}(p)$ for $i = 1, 2$,
  we have that both $\lp_{\omega_1}(qp)$ and $\lp_{\omega_2}(qp)$ contain at least two monomials.
  However, the only monomials of $qp$ that can appear in the leading part are $x^m$ and $y^n$,
  and there is a unique weight function so that they have the same weight.
  
  The second claim of the lemma follows from $\lp_{\omega}(q) \lp_{\omega}(p) = \lp_{\omega}(qp)$.
\end{proof}

There is an analogous version of Lemma~\ref{lem:necessary_condition} with the lowest homogeneous
part in place of the leading homogeneous part. However, even when both the lowest and the leading
homogeneous part are separable, the whole polynomial need not be separable, as the following
example shows. 

\begin{Example}
  \hangindent=-3cm\hangafter=-4
  For $p=(x^3+x^2y+xy^2+y^3)+y^2\in\set Q[x,y]$, the relevant weight 
  function for the leading homogeneous part as in Lemma~\ref{lem:necessary_condition} is given by
  $\omega_x=\omega_y=1$. It leads to the leading homogeneous part
  $x^3+x^2y+xy^2+y^3$. 
  Analogously, the relevant weight function for the  
  lowest homogeneous part is given by $\omega_x=2,\omega_y=3$.
  It leads to the lowest homogeneous part $x^3+y^2$. Both the leading
  and the lowest homogeneous part are separable.  
  We claim that $p$ is not separable. \hfill
  \smash{\raisebox{2.0em}{\llap{\quad%
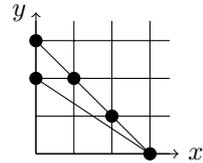
  \begin{tikzpicture}[scale=.5]
   \draw[->](0,0)--(3.75,0) node[right]{$x$};
   \draw[->](0,0)--(0,3.75) node[left]{$y$};
   \clip (-.5,-.5) rectangle (3.5,3.5);
   \draw (0,0) grid (4,4);
   \foreach \x/\y in {0/2,0/3,1/2,2/1,3/0} \fill (\x,\y) circle(5pt);
   \draw (0,3)--(3,0) (0,2)--(3,0);
  \end{tikzpicture}\kern.33em}}}

  \hangindent=0pt
  Let $\omega$ be the weight function defined by $\omega(x^iy^j)=2i+3j$, so that the 
  lowest homogeneous part of $p$ is $x^3+y^2$ (weight~6), and the next-to-lowest part is $x^2y$ (weight~7).
  With respect to $\omega$, any separated polynomial involving both variables 
  only consists of homogeneous parts $ax^n+by^m$ whose weight $2n=3m$ is a multiple of~6.
  
  Assume that $p$ is separable and let $q\in\set Q[x,y]\setminus\{0\}$ be such that $qp$ is separated.
  Write $q=q_0+q_1+\cdots$, where $q_0,q_1,\dots$ are the lowest, the next-to-lowest, etc.
  homogeneous parts of $q$ with respect to $\omega$.
  The lowest homogeneous part of $pq$ is then $q_0(x^3+y^2)$, and since it must be separated and 
  involve both variables, we have $\omega(q_0)=0\bmod6$.
  
  Because of $\omega(q_0x^2y)=\omega(q_0(x^3+y^2))+1=1\bmod6$, none of the terms 
  of $q_0x^2y$ can appear in $qp$, so they must all be canceled by something.
  We must therefore have $\omega(q_1)=\omega(q_0)+1$ and $q_0x^2y+q_1(x^3+y^2)=0$.
  This implies that $x^3+y^2$ divides~$q_0$, which in turn implies that the lowest
  homogeneous part $q_0(x^3+y^2)$ of $pq$ has a multiple factor. On the other hand,
  $q_0(x^3+y^2)=ax^n+by^m$ for some $a,b\neq0$, and every such polynomial is squarefree. 
  This is a contradiction. 
\end{Example}

The main result of the section is the following  ``partial converse'' of Lemma~\ref{lem:necessary_condition}.

\begin{Theorem}\label{thm:reduction_to_lp}
  Let $p \in \K[x, y] \setminus (\K[x] \cup \K[y])$ be a separable polynomial.
  Let $\omega$ be the weight function given by Lemma~\ref{lem:necessary_condition}, and let $P$
  be the minimal separated multiple of $p$. Then $\lp_{\omega}(P)$ is the minimal separated multiple of $\lp_{\omega}(p)$.
\end{Theorem}

Before proving the theorem, we will establish some combinatorial tools 
for dealing with divisors of separated polynomials extending the results
of Cassels~\cite{cassels}.

\begin{Notation}\label{notation:two_polys}
    Consider a separated polynomial $f(x) - g(y)$ with $\deg_x f = m$ and 
    $\deg_y g = n$, where $m, n > 0$, and a weight function $\omega(x^i y^j) = i n + j m$.
    We introduce a new variable $t$ and consider two auxiliary equations
    \[
      f(x) = t\quad \text{ and } \quad g(y) = t.
    \]
    We solve these equations with respect to $x$ and $y$ in $\overline{\K (t)}$, the algebraic closure of $\K (t)$.
    Let the solutions be $\alpha_0, \ldots, \alpha_{m - 1}$ and $\beta_0, \ldots, \beta_{n - 1}$, respectively.
    Then every element $\pi$ of $\Gal(\overline{\K (t)} / \K (t))$, the Galois group of $\overline{\K (t)}$ over $\K (t)$, acts on $\mathbb{Z}_m \times \mathbb{Z}_n$ by
    \[
        \pi (i, j) := (i', j')
        \iff
        (\pi(\alpha_i),\pi(\beta_j))=(\alpha_{i'},\beta_{j'}).
    \]
    Let $G \subseteq \mathbf{S}_m \times \mathbf{S}_n$ be the group of permutations induced on $\mathbb{Z}_m \times \mathbb{Z}_n$ by this action.
\end{Notation}

\begin{Notation}\label{notation:slice}
  For a subset $T \subseteq \mathbb{Z}_m \times \mathbb{Z}_n$, and $(i,j)\in\mathbb{Z}_m\times\mathbb{Z}_n$, we introduce
\[
    T_{i, \ast} := \{k \mid (i, k) \in T\} \text{ and } T_{\ast, j} := \{k \mid (k, j) \in T\}.
\]
\end{Notation}

\begin{Lemma}\label{lem:rows_cols_size}
  Let $T \subseteq \mathbb{Z}_m \times \mathbb{Z}_n$ be a $G$-invariant subset.
  Then $|T_{0, \ast}| = |T_{1, \ast}| = \ldots = |T_{m - 1, \ast}|$ and
  $|T_{\ast, 0}| = |T_{\ast, 1}| = \ldots = |T_{\ast, n - 1}|$.
\end{Lemma}

\begin{proof}
  We show that $|T_{0, \ast}| = |T_{1, \ast}|$, the rest is analogous.
  First, we observe that $f(x) - t$ is irreducible over $\K(t)$.
  If it was not, it would be reducible over $\K[t]$ due to Gauss's lemma.
  The latter is impossible because $f(x) - t$ is linear in $t$ and does not have factors in $\K[x]$.
  The irreducibility of $f(x) - t$ implies that its Galois group acts transitively on the roots.
  In particular, there exists $\pi \in \Gal(\overline{\K(t)} / \K(t))$ such that $\pi(\alpha_0) = \alpha_1$.
  Hence, $\pi$ maps $T_{0, \ast}$ to $T_{1, \ast}$, and we have $|T_{0, \ast}| \leqslant |T_{1, \ast}|$.
  The reverse inequality is analogous.
\end{proof}


\begin{Lemma}[{cf. \cite[p. 9-10]{cassels}}]\label{lem:Gset_divisors}
  Let $T \subseteq \mathbb{Z}_m \times \mathbb{Z}_n$ be a $G$-invariant subset.
  There exists a divisor $p$ of $f(x) - g(y)$, unique up to a multiplicative constant, such that 
  \begin{equation}\label{eq:p_property}
    T = \{(i, j) \in \mathbb{Z}_m \times \mathbb{Z}_n \mid p(\alpha_i, \beta_j) = 0\}.
   \end{equation}
\end{Lemma}

\begin{proof}
  \textit{Existence.}
  Let $T_{0, \ast} = \{ j_1, \ldots, j_s \}$.
  Since $f(\alpha_0) = t$, we have $\K (\alpha_0) \supseteq \K (t)$, 
  so every element of $\Gal(\overline{\K (t)} / \K (\alpha_0))$ leaves 
  $T$ invariant. If $\alpha_0$ is fixed, then $\beta_{j_1}, \ldots, \beta_{j_s}$ 
  are permuted.
  Therefore, the polynomial 
  $(y - \beta_{j_1})(y - \beta_{j_2}) \dots (y - \beta_{j_s})$
  is invariant under the action 
  of $\Gal(\overline{\K (t)} / \K (\alpha_0))$. 
  Hence, by the fundamental theorem of Galois theory, it is a polynomial in $\K (\alpha_0)[y]$.
  Since, by construction, it divides $f(\alpha_0) - g(y)$ over $\K (\alpha_0)$, 
  and $\alpha_0$ and $y$ are algebraically independent, 
  it in fact belongs to $\K [\alpha_0, y]$. 
  Replacing $\alpha_0$ by $x$, we find a polynomial $p\in \K [x, y]$, 
  which divides $f(x) - g(y)$ in $\K [x,y]$.

  Let $(i, j) \in \mathbb{Z}_m \times \mathbb{Z}_n$. 
  Since $\Gal(\overline{\K (t)} / \K (t))$ acts transitively on the roots of $f(x) - t$ (see the proof of Lemma~\ref{lem:rows_cols_size}), there is an automorphism $\pi$ 
  with $\pi (\alpha_i) = \alpha_0$.
  Let $\beta_{j'} = \pi(\beta_j)$.
  We then have
  \[
    p(\alpha_i, \beta_j) = 0 \iff p(\alpha_0, \beta_{j'}) = 0 \iff j' \in T_{0, \ast} \iff (i, j) \in T.
  \]

  \textit{Uniqueness.}
  It remains to prove that $p$ is unique up to a multiplicative constant. Assume that $\tilde{p}$ is another divisor of $f(x) - g(y)$ such that $\tilde{p}(\alpha_i,\beta_j) = 0$ for all $(i,j)\in T$. 
  The same argument which proved that $p$ is a divisor of $f(x) - g(y)$ applies to show that $p$ is a divisor of $\tilde{p}$ in $\K [x,y]$, and vice versa. Hence, they only differ by a multiplicative constant.
\end{proof}


\begin{Lemma}\label{lem:combinatorial_separatedness}
  Let $T \subseteq \mathbb{Z}_m \times \mathbb{Z}_n$ be a $G$-invariant subset.
  The unique factor $p$ corresponding to $T \subseteq \mathbb{Z}_m \times \mathbb{Z}_n$ (see Lemma~\ref{lem:Gset_divisors}) is separated if and only if
  \begin{equation}\label{eq:sep_condition}
        \forall\ i,j\in\mathbb{Z}_m : (T_{i, \ast} \cap T_{j, \ast} = \varnothing) \text{ or } (T_{i, \ast} = T_{j, \ast})
  \end{equation}
\end{Lemma}

\begin{proof}
  Assume that $T$ satisfies~\eqref{eq:sep_condition}, and let $T_{0, \ast} = \{j_1, \ldots, j_s\}$.
  Consider the corresponding polynomial $p$ constructed in the proof of Lemma~\ref{lem:Gset_divisors}, 
  which is of the form
  \[
  p(x,y) = y^s + a_{s - 1}(x) y^{s - 1} + \dots + a_0(x),
  \]
  where, for every $0 \leqslant i < s$ and $0 \leqslant j < m$, $a_i(\alpha_j)$ is (up to sign) the $s - i$-th elementary symmetric polynomial in $\{\beta_k \mid k \in T_{j, \ast}\}$.
  
  Since $p \mid f(x) - g(y)$, we have 
  $
  \lp_{\omega}(p) \mid \lp_{\omega} (f(x) - g(y)) = ax^m - by^n$, 
  with $a,b\in\K\setminus\{0\}$.
  Hence, $y^s$ belongs to $\lp_{\omega}(p)$, and so 
  $\omega(a_{i}(x)y^i) \leqslant \omega(y^s) = ms$ for all $i\in\{0,\dots,s-1\}$,
  This implies
  \[
  \deg_x a_i(x) \leqslant \frac{ms - mi}{n} = (s - i)\frac{m}{n}.
  \]
  
  Since $T$ is the disjoint union of the $T_{i,\ast}$'s and of the $T_{\ast,j}$'s, respectively, whose cardinality, by Lemma~\ref{lem:rows_cols_size}, does not depend on $i$ and $j$, and $T_{0, \ast}$, by definition, consists of $s$ elements, we find that $ms = |T| = n |T_{\ast,j_1}|$, in particular $\ell := |T_{\ast, j_1}| = \frac{ms}{n}$. 
  Hence there exist $0 = i_1 < i_2 < \ldots < i_\ell < m$ such that $j_1\in T_{i_1,\ast}\cap\ldots\cap T_{i_\ell, \ast}$ and so, by \eqref{eq:sep_condition}, $T_{i_1, \ast} = \ldots =  T_{i_\ell, \ast}$.
  This shows that the polynomial $a_j(x) - a_j(\alpha_0)$ has at least $\ell$ pairwise distinct roots, $\alpha_{i_1}, \ldots, \alpha_{i_\ell}$, while it has degree less than $\ell$ for $0 < j < s$.
  Hence, it is the zero polynomial, and $a_j(x)$ is a constant (which we denote by $a_j$).
  Therefore, $p$ is separated and of the form $p(x,y) = f_0(x) - g_0(y)$ with
  $f_0(x) = a_0(x)$ and $g_0(y) = -(y^s + a_{s-1}y^{s-1} + \dots + a_1 y)$.
  
  To prove the other implication, let $p(x,y) = f_0(x) - g_0(y)$ be a separated factor of $f(x) - g(y)$. It is sufficient to show that
  \[
  (i ,j),\; (i', j),\; (i, j')\in T \implies (i', j')\in T. 
  \]
  Indeed, $(i, j),\; (i', j)\in T$ implies that $f_0(\alpha_i) = f_0(\alpha_{i'})$, so that $f_0(\alpha_i) - g_0(\beta_{j'}) = 0$ implies that $f_0(\alpha_{i'}) - g_0(\beta_{j'}) = 0$, i.e. $(i', j') \in T$.
\end{proof}


Lemma~\ref{lem:combinatorial_separatedness} motivates the following definition.

\begin{Definition}
  \begin{enumerate}
       \item A subset $T \subseteq \mathbb{Z}_m \times \mathbb{Z}_n$ is called \emph{separated} if it satisfies~\eqref{eq:sep_condition}, that is
        \[
          \forall\ i,j\in\mathbb{Z}_m : (T_{i, *} \cap T_{j, *} = \varnothing) \text{ or } (T_{i, *} = T_{j, *}).
        \]
  
       \item The intersection of all separated subsets containing $T \subseteq \mathbb{Z}_m \times \mathbb{Z}_n$ is called the \emph{separated closure} of $T$ and denoted by $T^{\mathrm{sep}}$.
        Notice that the separated closure is separated.
  \end{enumerate}
\end{Definition}


\begin{Example}
\begin{enumerate}
\item
Let $f(x) = x^4$ and $g(y) = y^4 + 2 y^2 + 1$. The group of permutations on pairs of roots of $f(x) - t$ and $g(y) - t$ is generated by
$((0 1 2 3),(0 1 2 3)), ((0 3 2 1),(0 3)(1 2))$ and $(id,(0 2))$. According to $f(x) - g(y)$ having two separated irreducible factors, $x^2 - y^2 -1$ and $x^2 + y^2 + 1$, we find that there are two orbits, each of them forming a separated set (Figure~\ref{fig:fig2}).

    \item 
Let $f(x) - g(y) = x^6 - y^6$. 
Let $t^{1/6} \in \overline{\mathbb{C}(t)}$ be any 6th root of $t$, and let $\epsilon$ be a primitive 6th root of unity.
Then the polynomials $f(x) - t$ and $g(y) - t$ have the same roots, namely:
\[
  \alpha_i = \beta_i = \epsilon^i t^{1/6},\quad  i\in\{0,\dots,5\}.
\]
The Galois group of $\overline{\mathbb{C}(t)}$ permutes these elements cyclically,
so the induced action on $\mathbb{Z}_6^2$ is generated by $((0 1 2 3 4 5), (0 1 2 3 4 5))$.
Figure~\ref{fig:fig1} shows the sets $T$ for the various factors of $x^6-y^6$.
Observe that $T$ is separated if and only if the corresponding factor is separated.
Observe also that multiplying two factors corresponds to taking the union of the corresponding sets~$T$.
\end{enumerate}
\end{Example}

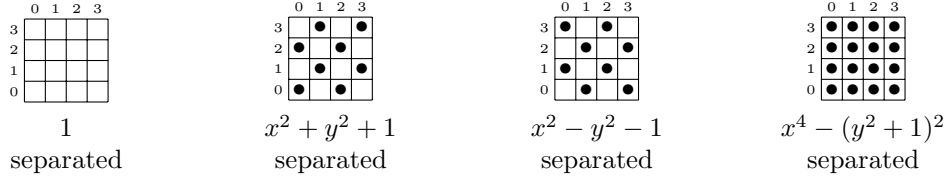
\begin{figure}
    \begin{center}
  \begin{tikzpicture}[scale=.275]
    \draw[xshift=-.5cm,yshift=-.5cm] (0,0) grid (4,4);
    \foreach\x in {0,1,2,3} \draw (-1,\x) node {\tiny$\x$} (\x,4) node {\tiny$\x$};
    \draw (1.5,-.75) node[below] {\vbox{\clap{$\mathstrut 1$}\kern0pt\clap{\strut separated}}};
  \end{tikzpicture}\hfil
  \begin{tikzpicture}[scale=.275]
    \draw[xshift=-.5cm,yshift=-.5cm] (0,0) grid (4,4);
    \foreach\x in {0,1,2,3} \draw (-1,\x) node {\tiny$\x$} (\x,4) node {\tiny$\x$};
    \foreach \x/\y in {0/0, 1/1, 2/2, 3/3, 2/0, 3/1, 0/2, 1/3} \draw (\x,\y) node {$\bullet$}; 
    \draw (1.5,-.75) node[below] {\vbox{\clap{$\mathstrut x^2+y^2+1$}\kern0pt\clap{\strut separated}}};
  \end{tikzpicture}\hfil
  \begin{tikzpicture}[scale=.275]
    \draw[xshift=-.5cm,yshift=-.5cm] (0,0) grid (4,4);
    \foreach\x in {0,1,2,3} \draw (-1,\x) node {\tiny$\x$} (\x,4) node {\tiny$\x$};
    \foreach \x/\y in {1/0, 2/1, 3/2, 0/3, 3/0, 0/1, 1/2, 2/3} \draw (\x,\y) node {$\bullet$}; 
    \draw (1.5,-.75) node[below] {\vbox{\clap{$\mathstrut x^2-y^2-1$}\kern0pt\clap{\strut separated}}};
  \end{tikzpicture}\hfil
  \begin{tikzpicture}[scale=.275]
    \draw[xshift=-.5cm,yshift=-.5cm] (0,0) grid (4,4);
    \foreach\x in {0,1,2,3} \draw (-1,\x) node {\tiny$\x$} (\x,4) node {\tiny$\x$};
    \foreach \x/\y in {0/0, 1/1, 2/2, 3/3, 2/0, 3/1, 0/2, 1/3} \draw (\x,\y) node {$\bullet$}; 
    \foreach \x/\y in {1/0, 2/1, 3/2, 0/3, 3/0, 0/1, 1/2, 2/3} \draw (\x,\y) node {$\bullet$}; 
    \draw (1.5,-.75) node[below] {\vbox{\clap{$\mathstrut x^4-(y^2+1)^2$}\kern0pt\clap{\strut separated}}};
  \end{tikzpicture}
\end{center}
    
    \caption{The factors of $x^4-(y^2+1)^2$ in $\set Q[x,y]$ and the sets $T\subseteq\set Z_4^2$ 
    corresponding to them.}
    \label{fig:fig2}
\end{figure}

\begin{figure}[h]
\begin{center}
\footnotesize
\begin{tikzpicture}[scale=.275]
  \draw[xshift=-.5cm,yshift=-.5cm] (0,0) grid (6,6);
  \foreach\x in {0,1,2,3,4,5} \draw (-1,\x) node {\tiny$\x$} (\x,6) node {\tiny$\x$};
  \draw (2.5,-.75) node[below] {\vbox{\clap{$\mathstrut 1$}\kern0pt\clap{\strut separated}}};
\end{tikzpicture}\hfil
\begin{tikzpicture}[scale=.275]
  \draw[xshift=-.5cm,yshift=-.5cm] (0,0) grid (6,6);
  \foreach\x in {0,1,2,3,4,5} \draw (-1,\x) node {\tiny$\x$} (\x,6) node {\tiny$\x$};
  \foreach\x/\y in {0/0, 1/1, 2/2, 3/3, 4/4, 5/5} \draw (\x,\y) node {$\bullet$}; 
  \draw (2.5,-.75) node[below] {\vbox{\clap{$\mathstrut x-y$}\kern0pt\clap{\strut separated}}};
\end{tikzpicture}\hfil
\begin{tikzpicture}[scale=.275]
  \draw[xshift=-.5cm,yshift=-.5cm] (0,0) grid (6,6);
  \foreach\x in {0,1,2,3,4,5} \draw (-1,\x) node {\tiny$\x$} (\x,6) node {\tiny$\x$};
  \foreach\x/\y in {0/3, 1/4, 2/5, 3/0, 4/1, 5/2} \draw (\x,\y) node {$\bullet$}; 
  \draw (2.5,-.75) node[below] {\vbox{\clap{$\mathstrut x+y$}\kern0pt\clap{\strut separated}}};
\end{tikzpicture}\hfil
\begin{tikzpicture}[scale=.275]
  \draw[xshift=-.5cm,yshift=-.5cm] (0,0) grid (6,6);
  \foreach\x in {0,1,2,3,4,5} \draw (-1,\x) node {\tiny$\x$} (\x,6) node {\tiny$\x$};
  \foreach\x/\y in {0/0, 1/1, 2/2, 3/3, 4/4, 5/5} \draw (\x,\y) node {$\bullet$}; 
  \foreach\x/\y in {0/3, 1/4, 2/5, 3/0, 4/1, 5/2} \draw (\x,\y) node {$\bullet$}; 
  \draw (2.5,-.75) node[below] {\vbox{\clap{$\mathstrut x^2-y^2$}\kern0pt\clap{\strut separated}}};
\end{tikzpicture}

\begin{tikzpicture}[scale=.275]
  \fill[lightgray] (-.5,2.5) rectangle (5.5, 3.5) (-.5,4.5) rectangle (5.5,5.5);
  \draw[xshift=-.5cm,yshift=-.5cm] (0,0) grid (6,6);
  \foreach\x in {0,1,2,3,4,5} \draw (-1,\x) node {\tiny$\x$} (\x,6) node {\tiny$\x$};
  \foreach\x/\y in {2/0, 3/1, 4/2, 5/3, 0/4, 1/5, 4/0, 5/1, 0/2, 1/3, 2/4, 3/5} \draw (\x,\y) node {$\bullet$}; 
  \draw (2.5,-.75) node[below] {\vbox{\clap{$\mathstrut x^2+xy+y^2$}\kern0pt\clap{\strut not separated}}};
\end{tikzpicture}\hfil
\begin{tikzpicture}[scale=.275]
  \draw[xshift=-.5cm,yshift=-.5cm] (0,0) grid (6,6);
  \foreach\x in {0,1,2,3,4,5} \draw (-1,\x) node {\tiny$\x$} (\x,6) node {\tiny$\x$};
  \foreach\x/\y in {0/0, 1/1, 2/2, 3/3, 4/4, 5/5} \draw (\x,\y) node {$\bullet$}; 
  \foreach\x/\y in {2/0, 3/1, 4/2, 5/3, 0/4, 1/5, 4/0, 5/1, 0/2, 1/3, 2/4, 3/5} \draw (\x,\y) node {$\bullet$}; 
  \draw (2.5,-.75) node[below] {\vbox{\clap{$\mathstrut x^3-y^3$}\kern0pt\clap{\strut separated}}};
\end{tikzpicture}\hfil
\begin{tikzpicture}[scale=.275]
  \fill[lightgray] (-.5,3.5) rectangle (5.5, 5.5);
  \draw[xshift=-.5cm,yshift=-.5cm] (0,0) grid (6,6);
  \foreach\x in {0,1,2,3,4,5} \draw (-1,\x) node {\tiny$\x$} (\x,6) node {\tiny$\x$};
  \foreach\x/\y in {0/3, 1/4, 2/5, 3/0, 4/1, 5/2} \draw (\x,\y) node {$\bullet$}; 
  \foreach\x/\y in {2/0, 3/1, 4/2, 5/3, 0/4, 1/5, 4/0, 5/1, 0/2, 1/3, 2/4, 3/5} \draw (\x,\y) node {$\bullet$}; 
  \draw (2.5,-.75) node[below] {\vbox{\clap{$\mathstrut x^3{+}2 x^2 y{+}2 x y^2{+}y^3$}\kern0pt\clap{\strut not separated}}};
\end{tikzpicture}\hfil
\begin{tikzpicture}[scale=.275]
  \fill[lightgray] (-.5,3.5) rectangle (5.5, 5.5);
  \draw[xshift=-.5cm,yshift=-.5cm] (0,0) grid (6,6);
  \foreach\x in {0,1,2,3,4,5} \draw (-1,\x) node {\tiny$\x$} (\x,6) node {\tiny$\x$};
  \foreach\x/\y in {0/0, 1/1, 2/2, 3/3, 4/4, 5/5} \draw (\x,\y) node {$\bullet$}; 
  \foreach\x/\y in {0/3, 1/4, 2/5, 3/0, 4/1, 5/2} \draw (\x,\y) node {$\bullet$}; 
  \foreach\x/\y in {2/0, 3/1, 4/2, 5/3, 0/4, 1/5, 4/0, 5/1, 0/2, 1/3, 2/4, 3/5} \draw (\x,\y) node {$\bullet$}; 
  \draw (2.5,-.75) node[below] {\vbox{\clap{$\mathstrut x^4{+}x^3 y{-}x y^3{-}y^4$}\kern0pt\clap{\strut not separated}}};
\end{tikzpicture}

\begin{tikzpicture}[scale=.275]
  \fill[lightgray] (-.5,2.5) rectangle (5.5, 3.5) (-.5,4.5) rectangle (5.5,5.5);
  \draw[xshift=-.5cm,yshift=-.5cm] (0,0) grid (6,6);
  \foreach\x in {0,1,2,3,4,5} \draw (-1,\x) node {\tiny$\x$} (\x,6) node {\tiny$\x$};
  \foreach\x/\y in {1/0, 2/1, 3/2, 4/3, 5/4, 0/5, 5/0, 0/1, 1/2, 2/3, 3/4, 4/5} \draw (\x,\y) node {$\bullet$}; 
  \draw (2.5,-.75) node[below] {\vbox{\clap{$\mathstrut x^2-xy+y^2$}\kern0pt\clap{\strut not separated}}};
\end{tikzpicture}\hfil
\begin{tikzpicture}[scale=.275]
  \fill[lightgray] (-.5,3.5) rectangle (5.5, 5.5);
  \draw[xshift=-.5cm,yshift=-.5cm] (0,0) grid (6,6);
  \foreach\x in {0,1,2,3,4,5} \draw (-1,\x) node {\tiny$\x$} (\x,6) node {\tiny$\x$};
  \foreach\x/\y in {0/0, 1/1, 2/2, 3/3, 4/4, 5/5} \draw (\x,\y) node {$\bullet$}; 
  \foreach\x/\y in {1/0, 2/1, 3/2, 4/3, 5/4, 0/5, 5/0, 0/1, 1/2, 2/3, 3/4, 4/5} \draw (\x,\y) node {$\bullet$}; 
  \draw (2.5,-.75) node[below] {\vbox{\clap{$\mathstrut x^3{-}2 x^2 y{+}2 x y^2{-}y^3$}\kern0pt\clap{\strut not separated}}};
\end{tikzpicture}\hfil
\begin{tikzpicture}[scale=.275]
  \draw[xshift=-.5cm,yshift=-.5cm] (0,0) grid (6,6);
  \foreach\x in {0,1,2,3,4,5} \draw (-1,\x) node {\tiny$\x$} (\x,6) node {\tiny$\x$};
  \foreach\x/\y in {0/3, 1/4, 2/5, 3/0, 4/1, 5/2} \draw (\x,\y) node {$\bullet$}; 
  \foreach\x/\y in {1/0, 2/1, 3/2, 4/3, 5/4, 0/5, 5/0, 0/1, 1/2, 2/3, 3/4, 4/5} \draw (\x,\y) node {$\bullet$}; 
  \draw (2.5,-.75) node[below] {\vbox{\clap{$\mathstrut x^3+y^3$}\kern0pt\clap{\strut separated}}};
\end{tikzpicture}\hfil
\begin{tikzpicture}[scale=.275]
  \fill[lightgray] (-.5,3.5) rectangle (5.5, 5.5);
  \draw[xshift=-.5cm,yshift=-.5cm] (0,0) grid (6,6);
  \foreach\x in {0,1,2,3,4,5} \draw (-1,\x) node {\tiny$\x$} (\x,6) node {\tiny$\x$};
  \foreach\x/\y in {0/0, 1/1, 2/2, 3/3, 4/4, 5/5} \draw (\x,\y) node {$\bullet$}; 
  \foreach\x/\y in {0/3, 1/4, 2/5, 3/0, 4/1, 5/2} \draw (\x,\y) node {$\bullet$}; 
  \foreach\x/\y in {1/0, 2/1, 3/2, 4/3, 5/4, 0/5, 5/0, 0/1, 1/2, 2/3, 3/4, 4/5} \draw (\x,\y) node {$\bullet$}; 
  \draw (2.5,-.75) node[below] {\vbox{\clap{$\mathstrut x^4{-}x^3 y{+}x y^3{-}y^4$}\kern0pt\clap{\strut not separated}}};
\end{tikzpicture}

\begin{tikzpicture}[scale=.275]
  \fill[lightgray] (-.5,3.5) rectangle (5.5, 5.5);
  \draw[xshift=-.5cm,yshift=-.5cm] (0,0) grid (6,6);
  \foreach\x in {0,1,2,3,4,5} \draw (-1,\x) node {\tiny$\x$} (\x,6) node {\tiny$\x$};
  \foreach\x/\y in {2/0, 3/1, 4/2, 5/3, 0/4, 1/5, 4/0, 5/1, 0/2, 1/3, 2/4, 3/5} \draw (\x,\y) node {$\bullet$}; 
  \foreach\x/\y in {1/0, 2/1, 3/2, 4/3, 5/4, 0/5, 5/0, 0/1, 1/2, 2/3, 3/4, 4/5} \draw (\x,\y) node {$\bullet$}; 
  \draw (2.5,-.75) node[below] {\vbox{\clap{$\mathstrut x^4+x^2 y^2+y^4$}\kern0pt\clap{\strut not separated}}};
\end{tikzpicture}\hfil
\begin{tikzpicture}[scale=.275]
  \fill[lightgray] (-.5,3.5) rectangle (5.5, 5.5);
  \draw[xshift=-.5cm,yshift=-.5cm] (0,0) grid (6,6);
  \foreach\x in {0,1,2,3,4,5} \draw (-1,\x) node {\tiny$\x$} (\x,6) node {\tiny$\x$};
  \foreach\x/\y in {0/0, 1/1, 2/2, 3/3, 4/4, 5/5} \draw (\x,\y) node {$\bullet$}; 
  \foreach\x/\y in {2/0, 3/1, 4/2, 5/3, 0/4, 1/5, 4/0, 5/1, 0/2, 1/3, 2/4, 3/5} \draw (\x,\y) node {$\bullet$}; 
  \foreach\x/\y in {1/0, 2/1, 3/2, 4/3, 5/4, 0/5, 5/0, 0/1, 1/2, 2/3, 3/4, 4/5} \draw (\x,\y) node {$\bullet$}; 
  \draw (2.5,-.75) node[below] {\vbox{\clap{$\mathstrut x^5{-}x^4 y{+}\cdots{-}y^5$}\kern0pt\clap{\strut not separated}}};
\end{tikzpicture}\hfil
\begin{tikzpicture}[scale=.275]
  \fill[lightgray] (-.5,3.5) rectangle (5.5, 5.5);
  \draw[xshift=-.5cm,yshift=-.5cm] (0,0) grid (6,6);
  \foreach\x in {0,1,2,3,4,5} \draw (-1,\x) node {\tiny$\x$} (\x,6) node {\tiny$\x$};
  \foreach\x/\y in {0/3, 1/4, 2/5, 3/0, 4/1, 5/2} \draw (\x,\y) node {$\bullet$}; 
  \foreach\x/\y in {2/0, 3/1, 4/2, 5/3, 0/4, 1/5, 4/0, 5/1, 0/2, 1/3, 2/4, 3/5} \draw (\x,\y) node {$\bullet$}; 
  \foreach\x/\y in {1/0, 2/1, 3/2, 4/3, 5/4, 0/5, 5/0, 0/1, 1/2, 2/3, 3/4, 4/5} \draw (\x,\y) node {$\bullet$}; 
  \draw (2.5,-.75) node[below] {\vbox{\clap{$\mathstrut x^5{+}x^4 y{+}\cdots{+}y^5$}\kern0pt\clap{\strut not separated}}};
\end{tikzpicture}\hfil
\begin{tikzpicture}[scale=.275]
  \draw[xshift=-.5cm,yshift=-.5cm] (0,0) grid (6,6);
  \foreach\x in {0,1,2,3,4,5} \draw (-1,\x) node {\tiny$\x$} (\x,6) node {\tiny$\x$};
  \foreach\x/\y in {0/0, 1/1, 2/2, 3/3, 4/4, 5/5} \draw (\x,\y) node {$\bullet$}; 
  \foreach\x/\y in {0/3, 1/4, 2/5, 3/0, 4/1, 5/2} \draw (\x,\y) node {$\bullet$}; 
  \foreach\x/\y in {2/0, 3/1, 4/2, 5/3, 0/4, 1/5, 4/0, 5/1, 0/2, 1/3, 2/4, 3/5} \draw (\x,\y) node {$\bullet$}; 
  \foreach\x/\y in {1/0, 2/1, 3/2, 4/3, 5/4, 0/5, 5/0, 0/1, 1/2, 2/3, 3/4, 4/5} \draw (\x,\y) node {$\bullet$}; 
  \draw (2.5,-.75) node[below] {\vbox{\clap{$\mathstrut x^6-y^6$}\kern0pt\clap{\strut separated}}};
\end{tikzpicture}
\end{center}
\caption{The factors of $x^6-y^6$ in $\set Q[x,y]$ and the sets $T\subseteq\set Z_6^2$ corresponding to them. For the unseparated cases, we highlight one choice of two incompatible rows.} 
\label{fig:fig1}
\vspace{-.25\baselineskip}
\end{figure}
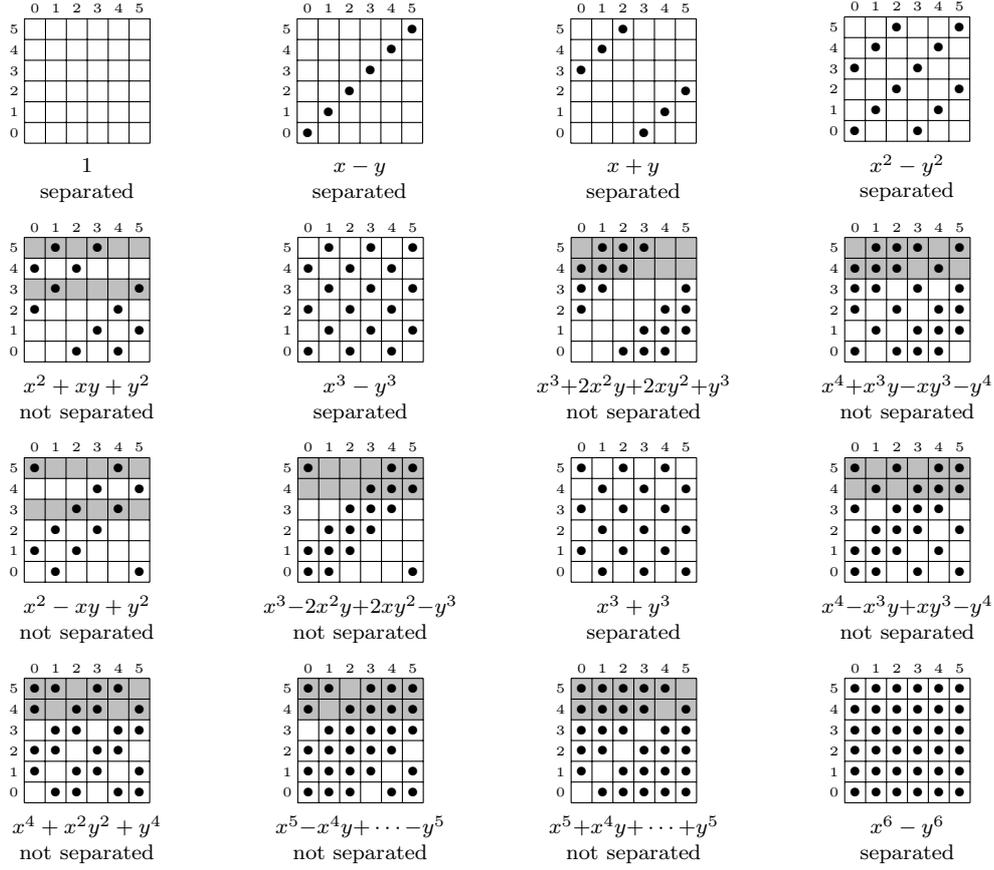


\begin{Lemma}\label{lem:closure_invariance}
    Let $T \subseteq \mathbb{Z}_m \times \mathbb{Z}_n$ be invariant with respect to $G \subseteq \mathbf{S}_m \times \mathbf{S}_n$.
    Then $T^{\sep}$ is also $G$-invariant.
\end{Lemma}

\begin{proof}
    Let $\pi = (\sigma, \tau) \in \mathbf{S}_m \times \mathbf{S}_n$, and let 
    $S \subseteq \mathbb{Z}_m \times \mathbb{Z}_n$ be a separated set.
    Since $\pi(S)_{i, \ast} = \tau(S_{\sigma(i), \ast})$,
    we find that $\pi(S)$ is separated as well.
    
    Assume that $T^{\sep}$ is not $G$-invariant, that is, there exists a $\pi \in G$ 
    such that $\pi(T^{\sep}) \neq T^{\sep}$.
    As we have shown, $\pi(T^{\sep})$ is separated, hence so is $S := T^{\sep} \cap \pi(T^{\sep})$.
    Observe that, since $\pi(T^{\sep}) \neq T^{\sep}$, $S \subsetneq T^{\sep}$.
    Since $T$ is $G$-invariant, $T \subseteq \pi(T^{\sep})$, so $T \subseteq S$.
    This contradicts the minimality of $T^{\sep}$.
\end{proof}


\begin{proof}[Proof of Theorem~\ref{thm:reduction_to_lp}]
  We use Notation~\ref{notation:two_polys} with $\overline{\K (t)}$ being identified with a subfield of the field $F$ of Puiseux series 
  in $t^{-1}$ over~$\overline{\K }$. 
  Let $\alpha_0, \ldots, \alpha_{m-1}$ and $\beta_0, \ldots, \beta_{n-1}$ denote the roots of $f(x)-t$ and $g(y) - t$ and $\overline{\alpha}_0, \ldots, \overline{\alpha}_{m - 1}$ and $\overline{\beta}_0, \ldots, \overline{\beta}_{n - 1}$ their highest degree terms.
  Observe that the highest degree terms are proportional to $t^{1/n}$ and $t^{1/m}$, and hence they are the roots of $\lp_{\omega}(f(x))- t$ and $\lp_{\omega}(g(y)) - t$, respectively.
  We define
  \begin{align*}
      T &= \{(i, j) \in \mathbb{Z}_m \times \mathbb{Z}_n \mid p(\alpha_i, \beta_j) = 0\},\\
      \overline{T} &= \{(i, j) \in \mathbb{Z}_m \times \mathbb{Z}_n \mid \lp_{\omega}(p)(\overline{\alpha}_i, \overline{\beta}_j) = 0\}.
  \end{align*}
  If $\lp_{\omega}(P)$ were not the minimal separated multiple of $\lp_\omega(p)$, by Lemma~\ref{lem:combinatorial_separatedness}, we would have $\overline{T}^{\sep} \subsetneq \mathbb{Z}_m \times \mathbb{Z}_n$.
  Therefore, it is sufficient to show that $\overline{T}^{\sep} = \mathbb{Z}_m \times \mathbb{Z}_n$.
  
  Since
  $$
  p(\alpha_i,\beta_j) = 0 \implies \lp_{\omega}(p)(\overline{\alpha_i},\overline{\beta_j}) = 0,
  $$
  we have $T\subseteq \overline{T}$. 
  By assumption, $P$ is the minimal separated multiple of $p$, so, by Lemma~\ref{lem:closure_invariance}, $T^{\sep} = \mathbb{Z}_m\times\mathbb{Z}_n$. 
  Since $T^{\sep}\subseteq \overline{T}^{\sep}$, this implies that
  $\overline{T}^{\sep} = \mathbb{Z}_m\times\mathbb{Z}_n$, and finishes the proof.
\end{proof}


\subsection{Algorithm}\label{subsec:algorithm}

The algorithm for finding a generator of the algebra of separated polynomials of a principal
ideal $\<p>$ is based on the results above. First, it uses Theorem~\ref{thm:reduction_to_lp} to 
reduce the situation to a homogeneous polynomial for a suitable grading, then, it uses
Proposition~\ref{prop:homog} to find a degree bound for the minimal separated multiple, and 
finally, it uses linear algebra to determine if such a multiple exists. 

\begin{Algorithm}\label{alg:1b}
  Input: $p\in\K[x,y]\setminus(\K[x]\cup\K[y])$. \\
  Output: $a\in\K[x]\times\K[y]$ such that $\K[a] = A(\langle p \rangle)$.
  The algorithm returns $a=(1,1)$ iff $A(\langle p \rangle)\cong\K$.
  
  \step 10 let $\omega_x,\omega_y\in\set N$ be maximal such that $p$ contains monomials $x^{\omega_y}y^0$ and $x^0y^{\omega_x}$. 
  Such parameters exist because $p$ is not univariate.
  \step 20 set $h=\lp_\omega(p)$ with $\omega(x^iy^j) := \omega_x i + \omega_y j$.
  \step 30 if $h$ does not contain $x^{\omega_y}$, return $(1,1)$. 
  \step 40 let $\{\zeta_1,\dots,\zeta_m\}\subseteq\overline{\K}$ be the roots of $h(x,1) \in \K[x]$.
  If any of them is not a simple root, return $(1,1)$.
  \step 50 let $N\in\set N$ be minimal such that $(\zeta_i/\zeta_j)^N=1$ for all $i,j$.
  If no such $N$ exists, return $(1,1)$.
  \step 60 make an ansatz 
  \[
  f=\sum_{i=0}^N a_i x^i,\quad g=\sum_{j=0}^{N\omega_x/\omega_y}b_jy^j,
  \]
  compute $\rem_x(f-g,p)$ in $\K(a_0,\dots,a_N,b_0,\dots,b_{N\omega_x/\omega_y},y)[x]$. The result of the reduction belongs to $\K[a_0,\dots,a_N,b_0,\dots,b_{N\omega_x/\omega_y},y,x]$ because 
  the leading coefficient of $p$ is in $\set K$.
  \step 70 equate the coefficients of $\rem_x(f-g,p)$ with respect to $x,y$ to zero and
  solve the resulting linear system for the unknowns $a_i,b_j$.
  \step 80 if there is a nonzero solution, return the corresponding pair $(f,g)$, otherwise return $(1,1)$.
\end{Algorithm}

When $\K$ is a number field, Step~5 can be carried out as follows:
for each ratio $\zeta_i/\zeta_j$, one should check whether the minimal polynomial of
this ratio over $\mathbb{Q}$
is a cyclotomic polynomial $\Phi_n$ and, if yes, return such $n$.
This check can be performed using a bound from~\cite[Theorem~15]{RS} that
yields the upper bound on $n$ based on the degree of the polynomial.

\begin{Proposition}
  Algorithm~\ref{alg:1b} is correct.
\end{Proposition}
\begin{proof}
  The algorithm consists of an application of the results of 
  the previous section and a handling of degenerate cases not covered
  by these results. In Steps 3--5, it is correct to return $(1,1)$ in
  the indicated situations because Proposition~\ref{prop:homog} implies that
  $h$ is not separable in these cases, which in combination with Lemma~\ref{lem:necessary_condition}
  implies that $p$ is not separable either. 
  
  By Proposition~\ref{prop:homog}, when $h$ has a separated multiple at all,
  it has one of weight~$N\omega_x$, and by Theorem~\ref{thm:reduction_to_lp}, 
  when $p$ has a separated multiple at all, it also has one of weight~$N\omega_x$. 
  Therefore, if $p$ has a separated multiple, 
  it will have one of the shape set up Step~6. For $f-g$ to be a separated
  multiple of $p$ is equivalent to $\rem_x(f-g,p)=0$, which we can safely 
  view as univariate division with respect to the variable $x$ because
  the leading coefficient of $p$ with respect to $x$ does not contain $y$
  (nor any of the undetermined coefficients). It is checked in Step~7
  whether there is a way to instantiate the undetermined coefficients
  in such a way that this remainder becomes zero. If so, any such way
  translates into a separated multiple, and by \cite[Theorem 2.3]{fried69},
  it is a generator of~$A(I)$. 
  If there is no non-zero solution, it is correct to return $(1,1)$.
\end{proof}


\section{Arbitrary bivariate Ideals}\label{sec:arbitrary}

The case of an arbitrary ideal $I\subseteq\K [x,y]$ is reduced to the two cases 
discussed in Sections~\ref{sec:0} and~\ref{sec:principal}. Every ideal
$I\subseteq\K[x,y]$ can be written as $I=\bigcap_{i=1}^kP_i$, where the $P_i$'s 
are primary ideals. Unless $I=\{0\}$ or $I=\<1>$, these primary ideals have
dimensions zero or one. Primary ideals in $\K[x,y]$ of dimension~1 must be principal
ideals, because $\dim(P_i)=1$ together with Bezout's theorem implies that $P_i$ cannot
contain any elements $p,q$ with $\gcd(p,q)=1$, and then $P_i$ being primary implies that
$P_i$ is generated by some power of an irreducible polynomial. 

The intersection of zero-dimensional ideals is zero-dimensional and the intersection
of principal ideals is principal, so there exists a zero-dimensional ideal $I_0$ and a principal
ideal $I_1$ such that $I=I_0\cap I_1$. These ideals are obtained as the intersections of
the respective primary components of~$I$. When $I_0=\<1>$ or 
$I_1=\<1>$, we have $I=I_1$ or $I=I_0$, respectively, and are in one of the cases 
already considered. Assume now that $I_1,I_0$ are both different from~$\<1>$. 

In order to use the results of Section~\ref{sec:principal}, we have to make sure that
the generator of $I_1$ contains both variables. If this is not the case, say if 
$I_1=\<h>$ for some $h\in\K[x]\setminus\K$, then 
the separated polynomials in $I$ are precisely the elements of $I\cap\K[x]$. 
If $p$ is such that $\<p>=I\cap\K[x]$, then the pairs $(x^ip,0)$ for $i=0,\dots,\deg_x p-1$ 
are generators of~$A(I)$ (see the proof of Proposition~\ref{prop:alg_zerodim_correct}), so this case is settled. 
Therefore, from now on we assume that the generator of $I_1$ 
contains both the variables.

We can compute generators of the algebra $A(I_0)\subseteq\K [x]\times\K [y]$
of separated polynomials in $I_0$ as described in Section~\ref{sec:0} and a
generator of the algebra $A(I_1)\subseteq\K [x]\times\K [y]$ of separated
polynomials in $I_1$ as described in Section~\ref{sec:principal}. Clearly, the algebra
$A(I)\subseteq\K [x]\times\K [y]$ of separated polynomials in $I$ is
$A(I)=A(I_0)\cap A(I_1)$. It thus remains to compute generators for this intersection. In
order to do so, we will exploit that the codimension of $A(I_0)$ as $\set
K$-subspace of $\K [x]\times\K [y]$ is finite (Lemma~\ref{lemma:1}), and that
$A(I_1)=\K [a]$ for some $a\in\K [x]\times\K [y]$. We have to find all polynomials~$p$
such that $p(a)\in A(I_0)$. Polynomials $p$ with a prescribed finite set of
monomials can be found with the help of Lemma~\ref{lemma:1} as follows.

\begin{Algorithm}\label{alg:2}
  Input: $a\in\K [x]\times\K [y]$, $A(I_0)$ and $V$ as in Lemma~\ref{lemma:1},
  and a finite set $S=\{s_1,\dots,s_m\}\subseteq\set N$.\\
  Output: a $\K $-vector space basis of the space of all polynomials $p$
  with $p(a)\in A(I_0)$ such that $p$ involves only monomials with exponents in~$S$.

  \step 10 for $i=1,\dots,m$, compute $r_i\in V$ such that $a^{s_i}-r_i\in A(I_0)$
  \step 20 compute a basis $B$ of the space of all $(c_1,\dots,c_m)\in\K ^m$
  with $c_1r_1+\cdots+c_mr_m=0$
  \step 30 for every element $(c_1,\dots,c_m)\in B$, return $c_1t^{s_1}+\cdots+c_mt^{s_m}$.
\end{Algorithm}

\begin{Proposition}
  Algorithm~\ref{alg:2} is correct.
\end{Proposition}
\begin{proof}
  If $(c_1,\dots,c_m)\in\K ^m$ is such that $\sum_{i=1}^m c_ia^{s_i}\in A(I_0)$,
  then $\sum_{i=1}^m c_ir_i\in A(I_0)$, and since $r_i\in V$ for all $i$ and $A(I_0)\cap V=\{0\}$,
  we have $\sum_{i=1}^m c_ir_i=0$. Therefore $(c_1,\dots,c_m)$ is among the vectors
  computed in step~2, so the algorithm does not miss any solutions.
  Conversely, if $(c_1,\dots,c_m)\in\K ^m$ is such that $\sum_{i=1}^m c_i r_i=0$,
  then $\sum_{i=0}^m c_ia^{s_i}=\sum_{i=0}^m c_i(a^{s_i}-r_i)\in A(I_0)$, so the algorithm
  does not return any wrong solutions.
\end{proof}

To find a set of generators of $A(I_0)\cap A(I_1)$, we apply Algorithm~\ref{alg:2} repeatedly.
First call it with $S=\{1,\dots,\dim V+1\}$. Since $|S|>\dim V$, the output must contain
at least one nonzero polynomial~$p_1$. If $d_1$ is its degree, we can restrict the search for
further generators to subsets $S$ of $\set N\setminus d_1\set N$, because when $q$ is such that
$q(a)\in A(I_0)$, then we can subtract a suitable linear combination of powers of $p_1$ to remove
from $q$ all monomials whose exponents are multiples of~$d_1$. 
When $d_1=1$, we have $A(I_0)\cap A(I_1)=\K [a]$ and are done. Otherwise,
$\set N\setminus d_1\set N$
is still an infinite set, so we can choose $S\subseteq\set N\setminus d_1\set N$
with $|S|>\dim V$ and call Algorithm~\ref{alg:2} to find another nonzero polynomial~$p_2$, say of
degree~$d_2$. The search for further generators can be restricted to polynomials consisting of
monomials whose exponents belong to $\set N\setminus(d_1\set N+d_2\set N)$. We can continue
to find further generators of degrees $d_3,d_4,\dots$ with $d_i\in\set N\setminus(d_1\set N+\cdots+d_{i-1}\set N)$
for all~$i$. Since the monoid $(\set N,+)$ has the ascending chain condition, this process
must come to an end.

The end is clearly not reached as long as $g:=\gcd(d_1,\dots,d_m)\neq1$, because then $\set N\setminus g\set N$ is an infinite subset of $\set N\setminus(d_1\set N+\cdots+d_m\set N)$.
Once we have reached $g=1$, it is well known~\cite{owens03,ramirez06} that $\set N\setminus(d_1\set N+\cdots+d_m\set N)$
is a finite set, and there are algorithms~\cite{beihoffer05} for computing its largest element
(known as the Frobenius number of $d_1,\dots,d_m$). We can
therefore constructively decide when all generators have been found.

Putting all steps together, our algorithm for computing the separated polynomials in an arbitrary
ideal of $\K [x,y]$ works as follows. We use the notation $\<d_1,\dots,d_m>$ for the submonoid
$d_1\set N+\dots+d_m\set N$ generated by $d_1,\dots,d_m$ in~$\set N$.

\begin{Algorithm}\label{alg:general}
  Input: an ideal $I\subseteq\K [x,y]$, given as a finite set of ideal generators\\
  Output: a finite set of generators for the algebra $A(I)$ of separated polynomials of~$I$

  \step 10 if $\dim I=0$, call Algorithm~\ref{alg:1}, return the result.
  \step 20 compute a zero-dimensional ideal $I_0$ and a principal ideal $I_1=\<h>$ with $I=I_0\cap I_1$ (for example, using Gr\"obner bases~\cite{becker93} and the remarks at the beginning of this section).
  \step 30 if $h\in\K[x]$, compute $p$ such that $\<p>=I\cap\K[x]$, return the pairs
    $(x^ip,0)$ for $i=0,\dots,\deg_x p-1$. Likewise if $h\in\K[y]$.
  \step 40 call Algorithm~\ref{alg:1} to get generators of~$A(I_0)$, and let $V$ be as in Lemma~\ref{lemma:1}.
  \step 50 call Algorithm~\ref{alg:1b} to get an $a\in\K [x]\times\K [y]$ 
    with $A(I_1)=\K [a]$. If $A(I_1)\cong\K $, return $(1, 1)$.
  \step 60 $G=\emptyset$, $\Delta=\emptyset$.
  \step 70 while $\gcd(\Delta)\neq1$, do:
  \step 81 select a set $S\subseteq\set N\setminus\<\Delta>$ with $|S|>\dim V$ and call Algorithm~\ref{alg:2}
  to find a nonzero polynomial $p$ with $p(a)\in A(I_0)$ consisting only of monomials with exponents in~$S$.
  \step 91 $G=G\cup\{p\}$, $\Delta=\Delta\cup\{\deg_x p\}$
  \step{10}0 call Algorithm~\ref{alg:2} with $S=\set N\setminus\<\Delta>$ (which is now a computable finite set)
  and add the resulting polynomials to~$G$.
  \step{11}0 return $G$
\end{Algorithm}

An implementation of the algorithm in Mathematica can be found on the website of the second author.
Incidentally, the algorithm also shows that $A(I)$ is always a finitely generated $\K $-algebra.

\begin{Example}
  For the ideal 
  \[
  I=\<(x^2-xy+y^2)(x^3-2xy^2-1),(x^2-xy+y^2)(y^3-2x^2y-1)>
  \]
  we have $I_0=\<x^3-2xy^2-1,y^3-2x^2y-1>$ and $I_1=\<x^2-xy+y^2>$.
  Algorithm~\ref{alg:1} yields a somewhat lengthy list of generators for~$A(I_0)$ from which 
  it can be read off that a suitable choice for $V$ is the $\K$-vector space generated
  by $(0,y^i)$ for $i=0,\dots,8$. In particular, $\dim V=9$.
  Algorithm~\ref{alg:1b} yields $A(I_1)=\K[(x^3,-y^3)]$.
  
  Making an ansatz for a polynomial $p$ of degree at most $10$ such that $p(a)\in A(I_0)$,
  we find a solution space of dimension~7. Its lowest degree element is $t^4-2t^2$,
  giving rise to the element $(x^{12}-2x^6,y^{12}-2y^6)$ of $A(I_0)\cap A(I_1)$. If we discard
  the other solutions and continue with the next iteration, we search for polynomials $p$ whose 
  support is contained $\{x^s:s\in S\}$ for $S=\{1,2,3,5,6,7,9,10,11,13\}$. Again, the solution 
  space turns out to have dimension~7.
  The lowest degree element is now $9t^5-26t^3+17$. Since $\gcd(4,5)=1$, we can exit
  the while loop. In step~10 of the algorithm, we get $S=\{1,2,3,6,7,11\}$, and this
  exponent set leads to a solution space of dimension three, generated by the polynomials
  $81t^6-323t^3$, $81t^7-539t^3+458$, and $6561t^{11}-191125t^3+184564$. The resulting
  generators of $A(I)=A(I_0)\cap A(I_1)$ are therefore the pairs $p((x^3,-y^3))$ where $p$ runs through
  the five polynomials found by the algorithm.
\end{Example}

\section{More than two variables}
\label{sec:more_than_two}

It is a natural question whether anything more can be said about the case of several
variables. Incidentally, a multivariate version would be needed in order to solve the 
combinatorial problem that motivated this research in the first place. 

Algorithm~\ref{alg:1} for bivariate zero-dimensional ideals works in the same way for zero-dimensional ideals of $\K[x_1,\dots,x_n,y_1,\dots,y_m]$ for arbitrary $n,m$.
Also Lemma~\ref{lemma:1} generalizes without problems.
We believe that with some further work, our results for
principal ideals can also be generalized to the case of several variables. 
However, in general, not every polynomial ideal with more than two variables is the 
intersection of a principal ideal and a zero-dimensional ideal, so the route taken in
Section~\ref{sec:arbitrary} is blocked. Also, as the next example shows 
we cannot expect an algorithm that finds the algebra of separated polynomials for an 
arbitrary ideal $I\subseteq\K[x_1,\dots,x_n,y_1,\dots,y_m]$, 
since it does not need to be finitely generated.

\begin{Example}[$A(I)$ is not necessarily finitely generated]\label{ex:not_finitely_generated}
  It is shown in \cite[Example 1.3]{mondal17} that the algebra
  \begin{alignat*}1
    R &:= \set C[t_1^2,t_1^3,t_2^{\vphantom2}]\cap\set C[t_1^2,t_2^{\vphantom2}-t_1^{\vphantom2}] \subset \set C[t_1, t_2]
  \end{alignat*}
  is not finitely generated. Consider the ideal
  \begin{alignat*}1
    I &= \langle x_1-t_1^2,x_2-t_1^3,x_3-t_2^{\vphantom2},\\
      &\qquad y_1-t_1^2,y_2-(t_2-t_1)\rangle\cap\set C[x_1,x_2,x_3,y_1,y_2]
    \\&=\<x_1 - y_1, -x_2 + x_3 y_1 - y_1 y_2, x_3^2 - y_1 - 2 x_3 y_2 + y_2^2>.
  \end{alignat*}
  We claim that $A(I)\cong R$ as $\set C$-algebras, implying that $A(I)$ is
  not finitely generated. 
  We show that $\phi\colon A(I)\to R$ defined by  $\phi(f,g)=f(t_1^2,t_1^3,t_2^{\vphantom2})$ is an isomorphism:
  \begin{itemize}
  \item $\phi$ is well-defined (the image is contained in~$R\subseteq\set C[t_1^2,t_1^3,t_2^{\vphantom2}]$).
    To see this, note that, $(f,g)\in A(I)$ means $f-g\in I$, which by definition of $I$ means
    $f(t_1^2,t_1^3,t_2^{\vphantom2})=g(t_1^2,t_2-t_1)$. Therefore, $f(t_1^2,t_1^3,t_2^{\vphantom2})\in\set C[t_1^2,t_2^3,t_2^{\vphantom2}]\cap\set C[t_1^2,t_2-t_1]=R$.
  \item $\phi$ is surjective. For every $p\in R$ there exist polynomials $f,g$ with $p=f(t_1^2,t_1^3,t_2^{\vphantom2})=g(t_1^2,t_2-t_1)$.
    By definition of $I$ we have $f(x_1,x_2,x_3)-g(y_1,y_2)\in I$, hence $(f,g)\in A(I)$.
    Now $\phi(f)=p$, so $p$ is in the image of~$\phi$.
  \item $\phi$ is injective. This follows from $I\cap\set C[y_1,y_2]=\{0\}$. \qed
  \end{itemize}
\end{Example}

It would still make sense to ask for an algorithm that decides whether $A(I)$ is nontrivial. 
We do not have such an algorithm, but being able to solve the problem in the bivariate case gives rise to a 
necessary condition.

\begin{Proposition}
 Let 
 \[\xi\colon \mathbb{K}[x_1, \ldots, x_n] \to \mathbb{K}[x] \;\text{ and }\;\eta\colon \mathbb{K}[y_1, \ldots, y_m] \to \mathbb{K}[y]
 \]
 be two homomorphisms, and let $I\subseteq \mathbb{K}[x_1,\dots,x_n,y_1,\dots,y_m]$ be an ideal such that 
 \[ 
 I\cap \mathbb{K}[y_1,\dots,y_m] = \{0\} \;\text{ and }\;(\mathrm{id}\otimes\eta)(I)\cap \mathbb{K}[x_1,\dots,x_n] = \{0\}.
 \]
 If the algebra of separated polynomials of $I$ is non-trivial, then so is the algebra of separated polynomials of $J:= (\xi \otimes \eta) (I)\subseteq\set K[x, y]$.
\end{Proposition}

\begin{proof}
  Let $(f,g)$ be an arbitrary, non-constant element of $A(I)$. If
  $(\xi(f),\eta(g))\in A(J)$ were a $\K$-multiple of $(1,1)$, 
  we would find that $f - \eta(g)$ were an element of $(\mathrm{id}\otimes\eta)(I)\cap
  \mathbb{K}[x_1,\dots,x_n]$, and hence, by our assumption, that $f$ itself were 
  a constant. So $f-g\in I\cap\mathbb{K}[y_1,\dots,y_m]$, and hence, by 
  assumption, $g=f$ is a constant as well, contradicting that $(f,g)$ is not a constant. 
\end{proof}

The examples below show different reasonable choices for homomorphisms $\xi$ and $\eta$.

\begin{Example}
  Consider the polynomial $p = x^2 + x y_1 y_2 + y_1^2 + y_2^2$. 
  Let $\xi = \mathrm{id}$ and let $\eta$ be defined by $\eta(y_1) = y$, $\eta(y_2) = 2$.
  Notice that $\eta$ is just the evaluation of $y_2$ at $2$.
  Then $(\xi\otimes\eta)(p) = x^2 + 2 x y_1 + y_1^2 + 4$, a polynomial that is not separable. 
  Hence $p$ is not separable.
\end{Example}

\begin{Example}
  Consider the polynomial $p = x^2 + xy_1 + y_1^2 + y_2^4$.
  We cannot use the same strategy as in the previous example because any evaluation of $y_1$ or $y_2$ results in a separable polynomial. Nevertheless, the homomorphism defined by $\xi(x) = x$, $\eta(y_1) = y^2$, and $\eta(y_2) = y$
  maps $p$ to $(\xi \otimes \eta)(p) = x^2 + xy^2 + 2y^4$, a polynomial which is not separable.
  So $p$ is not separable either.
\end{Example}

\par\medskip\noindent\textbf{Acknowledgements.}
We thank Erhard Aichinger and Josef Schicho for sharing their thoughts on the
topic and for providing pointers to the literature. 
We also thank the referees for their careful reading and their valuable suggestions.
MB was supported by the Austrian FWF grant F5004. Part of this work was done during the visit of MB to HSE University. MB would like to thank the Faculty of Computer Science of HSE for its hospitality.
MK was supported by the Austrian FWF grants F5004 and P31571-N32.
GP was supported by NSF grants CCF-1564132, CCF-1563942, DMS-1853482, DMS-1853650, and DMS-1760448, by
PSC-CUNY grants \#69827-0047 and \#60098-0048.

\bibliographystyle{abbrvnat}
\bibliography{main}

\end{document}